\documentclass{iopart}
\usepackage{hyperref}
\usepackage{graphicx}
\usepackage{iopams}
\usepackage{amsthm}

\newcommand{\arxiv}[1]{\href{http://arxiv.org/abs/#1}{arXiv:#1}}

\newtheorem{theorem}{Theorem}[section]

\newtheorem{lemma}[theorem]{Lemma}

\newtheorem{remark}[theorem]{Remark}

\newcommand{\R}{{\mathbb R}}
\newcommand{\N}{{\mathbb N}}
\newcommand{\Z}{{\mathbb Z}}
\newcommand{\C}{{\mathbb C}}

\newcommand{\be}{\begin{equation}}
\newcommand{\ee}{\end{equation}}
\newcommand{\bea}{\begin{eqnarray}}
\newcommand{\eea}{\end{eqnarray}}

\newcommand{\ol}{\overline}

\newcommand{\what}{\widehat}

\newcommand{\I}{\mathrm{i}}
\newcommand{\E}{\mathrm{e}}

\newcommand{\cn}{\mathrm{cn}}
\newcommand{\sn}{\mathrm{sn}}
\newcommand{\dn}{\mathrm{dn}}
\newcommand{\sech}{\mathrm{sech}}
\newcommand{\sinc}{\mathrm{sinc}}

\newcommand{\eps}{\varepsilon}

\begin{document}

\title{Cnoidal Waves on Fermi--Pasta--Ulam Lattices}

\author{G.~Friesecke and A.~Mikikits-Leitner}

\address{Center for Mathematics, TU Munich, Boltzmannstrasse 3, 85748 Garching bei M\"unchen, Germany}
\ead{\mailto{gf@ma.tum.de}, \mailto{mikikits@ma.tum.de}}


\date{\today}

\begin{abstract}
We study a chain of infinitely many particles coupled by nonlinear springs, obeying the equations of motion
\[ 
   \ddot{q}_n = V'(q_{n+1}-q_n) - V'(q_n-q_{n-1})
\]
with generic nearest-neighbour potential $V$. We show that this chain carries exact spatially periodic travelling waves whose profile is asymptotic,
in a small-amlitude long-wave regime, to the KdV cnoidal waves. 
The discrete waves have three interesting features: (1) being exact travelling waves they keep their shape for infinite time, rather than just up to a timescale of order wavelength$^{-3}$ suggested by formal asymptotic analysis, (2) unlike solitary waves they carry a nonzero amount of energy per particle, (3) analogous behaviour
of their KdV continuum counterparts suggests long-time stability properties under nonlinear interaction with each other. Connections with the
Fermi-Pasta-Ulam recurrence phenomena are indicated. Proofs involve an adaptation of the renormalization approach of \cite{friesecke1999solitary} to a periodic setting
and the spectral theory of the periodic Schr\"odinger operator with KdV cnoidal wave potential. 

\end{abstract}


\section{Introduction}

The Fermi-Pasta-Ulam model consists of a one-dimensional chain of particles coupled by nonlinear springs, obeying the equations of motion
\begin{equation} \label{equ:FPUintro}
   \ddot{q}_n = V'(q_{n+1}-q_n) - V'(q_n-q_{n-1}). 
\end{equation}
Here $q_n(t)$ is the displacement of the n$^{th}$ particle out of equilibrium at time $t$, and $V$ is an anharmonic potential such as 
$V(\phi)=\phi ^2/2+\alpha \phi^3/3$ with $\alpha\neq 0$ (the FPU-$\alpha$ chain) or $V(\phi)=\phi ^2/2+\beta \phi^4/4$ with $\beta\neq 0$ (the FPU-$\beta$ chain). This model provides a fascinating paradigm of nonlinear Hamiltonian many-particle dynamics. On the one hand, it is simple enough to allow insight by rigorous
mathematical analysis. On the other hand, it already exhibits a rich spectrum of phenomena of wider importance:
coherent signal- and energy transport (as seen in biomolecules such as DNA); near-integrable behaviour (as documented by 
the fact that in certains regimes, the FPU model is well approximated
by the Korteweg-de Vries equation, of which more below); dispersive shocks (as seen in molecularly resolved gas dynamics);
 and statistical irreversibility and thermalization effects despite microscopic reversibility (as described by statistical
mechanics).

A central numerical phenomenon in the system (\ref{equ:FPUintro}) is a crossover from energy trapping in a few long-wave modes at low initial energy per particle to ergodic-like spreading of energy to short-wave modes at high initial energy per particle. See Fermi, Pasta, Ulam and Tsingou\footnote{Mary Tsingou was involved in the numerical work as acknowledged in the original report.} \cite{fermi1955studies} for first observations at low energy, Israilev and Chirikov \cite{izrailev1966statistical} for first observations at high energy, Dreyer and Herrmann \cite{dreyer2008numerical} for energy transfer to short-wave modes via dispersive shocks, and \cite{biello2002stages} for a nice review.    

A significant amount of understanding of the recurrent and non-statistical behaviour at low energy has emerged via approximation of the FPU model by completely integrable systems. In this paper \\
-- we argue that this level of understanding, reviewed below, is not completely satisfactory because it fails to cover the physically most interesting regime,
fixed nonlinearity, nonzero energy per particle, and larger and larger system size \\
-- we hope to convince readers that a deeper understanding could come via establishing existence, and long-time stability under interaction with each other, of spatially periodic and quasi-periodic waves in the infinite FPU chain \\
-- and we rigorously carry out a first step, establishing existence of certain spatially periodic travelling waves in FPU which are good candidates for such special
interaction properties, namely waves which are, in a small-amplitude long-wave limit, asymptotic to the celebrated KdV cnoidal waves. 

In the remainder of this Introduction we review current theoretical understanding of non-statistical behaviour of FPU at low energy, and informally describe
our results. 

{\bf 1. The invariant tori explanation.} One line of thought, going back to \cite{izrailev1966statistical}, is to consider a fixed and finite number of particles only, and linearize, i.e.~approximate the anharmonic interaction potential by a purely harmonic one. The resulting system is, of course, completely integrable, with
the phase space being foliated by the invariant tori given by the set of states with a given fixed amount of energy in each normal mode. 
By Kolmogorov-Arnol'd-Moser (KAM) theory, when switching on the anharmonic terms, many invariant tori survive. This prevents typical solutions from spreading energy to the entire phase space. This argument can be made rigorous by using Birkhoff normal forms \cite{rink2006proof}.

{\bf 2. The soliton explanation.}
Another approach, going back to Zabusky and Kruskal \cite{zabusky1965interaction}, is to keep the number of particles infinite, consider a suitable small-amplitude
long-wave regime, and observe (see Remark~\ref{rem:contlimes} below for their heuristic argument) that (\ref{equ:FPUintro}) is then well-approximated by a \emph{nonlinear, infinite-dimensional} integrable equation, the KdV equation
\begin{equation} \label{equ:KdVZK}
u_t+12\frac{V'''(0)}{V''(0)}uu_x+u_{xxx}=0.
\end{equation}
(More precisely, the latter approximates the FPU-$\alpha$ chain; in case of the FPU-$\beta$ chain one obtains the mKdV equation.)
The explanation of non-statistical behaviour advocated in \cite{zabusky1965interaction} then went as follows: the latter equation carries 'solitons', i.e.~solitary waves which propagate exactly under the nonlinear KdV dynamics; numerically these solitons also
show persistent shapes and velocities after nonlinearly interacting with each other; and general spatially localized initial data can be well approximated by
superposition of finitely many such solitons. Hence, again, energy is prevented from spreading to the entire phase space.
Significant aspects of this picture has nowadays become rigorous mathematics. In particular, the \emph{inverse scattering transform} introduced in 1974 by Gardner et al.~\cite{gardner1974korteweg} allows to prove the radiationless interaction of solitons, as well as show that all spatially localized initial data
asymptotically split up into a superposition of solitons (see e.g. \cite{grunert2009longtime} for more information and further references). And the approximation 
of (\ref{equ:FPUintro}) by (\ref{equ:KdVZK}) can be made rigorous up to timescales of order (wavelength)$^{3}$ for general localized solutions \cite{wayne2000counter},
and globally in time for special solutions of soliton type \cite{friesecke1999solitary, friesecke2002solitary, friesecke2004solitary3, friesecke2004solitary4} and 
multi-soliton type \cite{mizumachi2011nsoliton, mizumachi2009asymptotic}. 

From a physical point of view, neither the invariant tori explanation nor the soliton explanation are completely satisfactory. This is because they are limited to regimes satisfying certain undesirable restrictions:
\begin{itemize}
\item The invariant tori explantion is based on finite-dimensional KAM theory, and hence the allowed size of the anharmonicity tends rapidly to zero as the particle number gets large \cite{wayne1984kam}. It thus does not apply to the natural situation of many particles interacting via a fixed nonlinear potential. 
Very interesting results are available on KAM theory for perturbations of infinite-dimensional systems \cite{kappeler2003kdv}, but perturbing the harmonic lattice
(or indeed other infinite-dimensional integrable systems like the Toda lattice or KdV) into FPU lies well beyond the scope of these results.  
\item The soliton explanation does apply to infinitely many particles interacting via a fixed nonlinear potential, 
but it is only valid for spatially localized states; but spatial localization forces the energy per particle to be 
zero.\footnote{Curiously, the numerical KdV simulations in initial paper \cite{zabusky1965interaction} on the
soliton explanation were done at finite energy per particle, via imposing periodic boundary conditions, but the subsequent
theoretical analysis of soliton interactions was done at zero energy per particle, via just considering a fixed
number of localized solitons on the whole real line.}  
\end{itemize}

How, then, could we access the behaviour of infinite FPU chains with fixed nonlinear potential and nonzero energy per particle? Our own, 
admittedly rather modest, contribution towards this challenge is the following. We show that (\ref{equ:FPUintro}) carries exact spatially periodic travelling waves
\begin{equation} \label{resultone}
   r_n(t) = R(n-ct)
\end{equation}
with small-amplitude long-wave profile
\begin{equation} \label{resulttwo}
   R(n-ct) = \eps^2 \Phi(\eps(n-ct)) + O(\eps^4)
\end{equation}
and near-sonic velocity
\begin{equation} \label{resultthree}
   c^2 = V''(0)(1+\frac{\eps^2}{12}),
\end{equation}
where $r_n(t)=q_{n+1}(t)-q_n(t)$ is the relative displacement between neighbouring particles and $\Phi$ is a KdV cnoidal wave profile of speed 1, i.e.~a spatially periodic function such that $u(x,t)=\Phi(x-t)$ solves the KdV equation (\ref{equ:KdVZK}). 
Such solutions to KdV exist and are known explicitly via algebraic and geometric methods, see Figure~\ref{fig:cnoidalcconst} and the end of
this Introduction. 

\begin{figure}
\centering
\includegraphics[width=.8\textwidth]{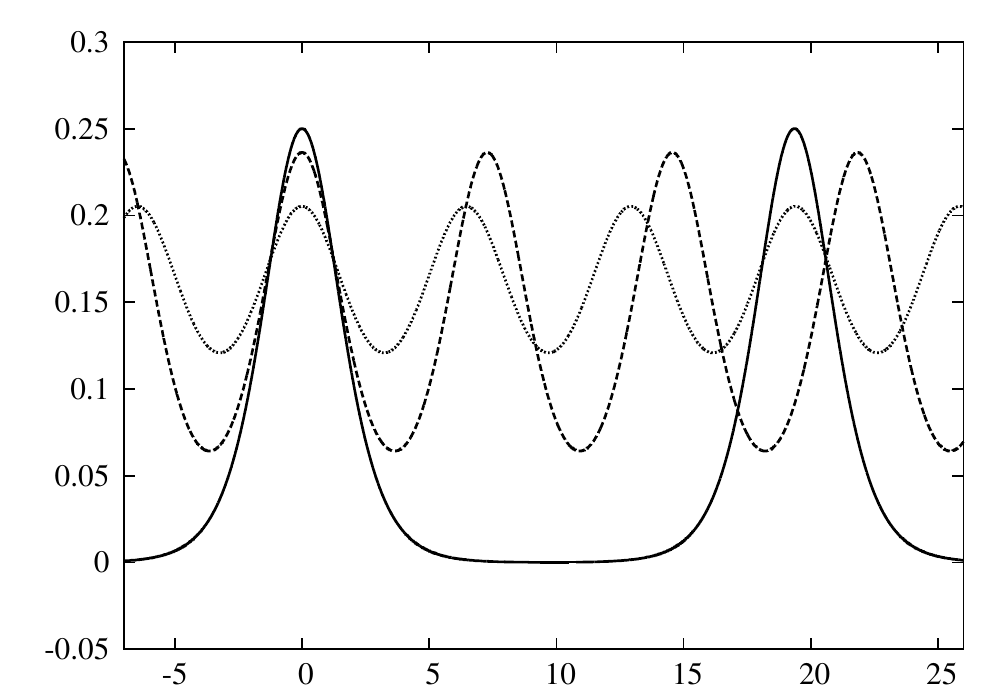}
\caption{Different cnoidal wave solutions of \eref{equ:KdVZK}, all moving with speed one.
We have set $V''(0)=V'''(0)=1$. The three curves represent the solutions given explicitly by (\ref{equ:cnoidalsolkL}), with elliptic modulus $k^2=0.3$ (dotted), $k^2=0.6$ (dashed), and $k^2=0.999$ (solid). The corresponding values of the period are chosen to make the wavespeed equal to one, 
i.e.~$2L=4K(k)(1-k^2+k^4)^{1/4}$, and are approximately $6.5$, $7.3$, and $19.4$.\label{fig:cnoidalcconst}}
\end{figure}

While interesting non-explicit existence results on periodic FPU travelling waves have
been obtained previously \cite{valkering1978periodic, filip1999existence, pankov2000travelling, herrmann2010unimodal},
the key point which makes the new result (\ref{resultone})--(\ref{resultthree}) promising in the context of 
energy trapping in FPU at low energy per particle is now the following. The discrete cnoidal waves 
(\ref{resultone})--(\ref{resultthree}) are {\it not arbitrary FPU travelling waves with finite energy per particle}, 
but are good candidates for exhibiting {\it special stability properties under nonlinearly interacting with each other}, on account of analogous remarkable properties of their continous KdV counterparts. As made more precise below, the KdV cnoidal 
waves are the spatially periodic analoga of the KdV solitons, and appropriate nonlinear superpositions, the so-called finite-gap KdV solutions, are (spatially quasi-periodic) analoga of the KdV multi-solitons. A long-term goal, then, which lies well beyond the scope and tools of this article, would be to investigate existence and long-time stability of discrete FPU-analoga of general finite-gap KdV solutions. Perhaps infinite-time
stability of such waves in the spirit of the deep recent results of Mizumachi~\cite{mizumachi2011nsoliton, mizumachi2009asymptotic} on discrete FPU-analoga of multi-solitons no longer holds, but it is conceivable that long finite-time results analogous to those of Hoffman and Wayne~\cite{hoffman2008counter} on FPU solitary wave interactions can be transferred to a quasi-periodic setting. 

The above line of thought also suggests, as to the best of our knowledge has not been pointed out previously, that the linear normal modes which were nonlinearly evolved
in the original FPU simulations \cite{fermi1955studies} and underlie the ''KAM explanation'' \cite{rink2006proof} enjoy special interaction properties under nonlinear
FPU dynamics. This is because, in the limit of low amplitude, the KdV cnoidal waves become close to linear waves plus 
a constant; see the dotted curve in Figure~\ref{fig:cnoidalcconst} and (\ref{equ:linearlimit}) in Section~\ref{sec:KdVsolutions}. 

Next, let us compare our result (\ref{resultone})--(\ref{resultthree}) to
previous work on localized waves in the FPU model. Existence of travelling waves with constant
asymptotic values at infinity has been established via different approaches: variational 
methods \cite{friesecke1994existence, smets1997solitary, filip1999existence, pankov2000travelling, 
herrmann2010unimodal}, see also \cite{schwetlick2007solitary, herrmann2012subsonic} for a generalization to 
nonconvex potentials, center manifold arguments \cite{iooss2000travelling}, and comparison to KdV 
\cite{friesecke1999solitary}. The latter two approaches are limited to small amplitude, but also deliver the waveform. 

We close this Introduction by describing briefly the above-mentioned finite-gap KdV solutions, and
explaining how they encode and reveal remarkable nonlinear interaction properties of the KdV cnoidal waves
in Figure~\ref{fig:cnoidalcconst}. The finite-gap solutions can be constructed via algebraic and geometric methods going back to
Its, Matveev, Dubrovin, McKean, and van Moerbeke \cite{its1975hill, its1975schroedinger, dubrovin1975periodic, dubrovin1975inverse, mckean1975spectrum}, and are given by
\begin{eqnarray} 
&&u(x,t)=c+2\frac{\partial^2}{\partial x^2}\ln \Theta(\mathbf{\xi}), \label{equ:itsmatveevformula}\\
&&\Theta(\mathbf{\xi};B)=\sum_{\mathbf{k}\in\Z^N}\exp \left\{ 2\pi\I \langle \mathbf{k},\mathbf{\xi}\rangle+\pi \I \langle \mathbf{k},B\mathbf{k}\rangle\right\} ,\nonumber
\end{eqnarray}
where $\langle .,.\rangle$ denotes the scalar product in $\R^N$ and $c$ is a constant. Moreover, $B$ is a symmetric matrix and $\mathbf{\xi}=(\xi_1,\dots,\xi_N)$ denotes the vector of the generalized phases $\xi_j=K_jx-\omega_jt+\phi_j$, $j=1,\dots,N$, which are determined by an underlying prescribable Riemann surface (see~\ref{app:periodicKdVsol}).   
In general, the $K_j$'s and $\omega_j$'s are incommensurable quantities and thus KdV solutions of the form \eref{equ:itsmatveevformula} are quasi-periodic in $x$ and $t$. 
In the one-gap case ($N=1$) the solutions \eref{equ:itsmatveevformula} reduce to the, spatially periodic, KdV cnoidal waves given by 
\begin{equation} \label{equ:cnoidalwave}
u(x,t)=E_2+(E_3-E_2)\cn^2\Big( \sqrt{\frac{E_3-E_1}{2}}\big(x-2(E_1+E_2+E_3)t\big);k^2\Big),
\end{equation}
where $E_1<E_2<E_3$ and $k^2=(E_3-E_2)/(E_3-E_1)$ is the elliptic modulus. Here $\cn$ denotes one of the Jacobi elliptic functions. See Figure \ref{fig:cnoidalcconst}. The solutions \eref{equ:cnoidalwave} were already known to Korteweg and de~Vries in 1895 \cite{korteweg1895change}.

Physically, the general finite-gap solutions \eref{equ:itsmatveevformula} can be interpreted as a linear superposition of $N$ cnoidal waves of the form \eref{equ:cnoidalwave} plus a nonlinear interaction term. Here the diagonal elements of the symmetric matrix $B$ determine the periods of these cnoidal waves, while the off-diagonal elements give their nonlinear interaction. In this precise sense, these solutions are finite-energy-per-unit-volume analoga of the KdV multi-solitons.

For a more detailed description of quasi-periodic KdV solutions we refer to~\ref{app:periodicKdVsol}.

\section{Main result} \label{sec:result}
We now describe precisely or result on existence and shape of cnoidal-type waves in Fermi-Pasta-Ulam chains. 

The governing equations are 
\begin{equation} \label{equ:equofmot}
\partial_t^2q(j,t)=V'(q(j+1,t)-q(j,t))-V'(q(j,t)-q(j-1,t)) \;\;\; (j\in\Z),
\end{equation}
where $q(j,t)$ denotes the displacement out of equilibrium of the $j$th particle. Formally, the associated Hamiltonian energy
\begin{equation}
H(t)=\sum_{j\in \Z} \left( \frac{p(j,t)^2}{2}+V(q(j+1,t)-q(j,t))\right)
\end{equation}
(with $\partial_tq(j,t)=:p(j,t)$ denoting the particle momenta)
is conserved along solutions, but our interest is in infinite-energy solutions.
In this paper we make the following assumptions on the nearest-neighbour potential:
\begin{equation} \label{equ:hypopotential}
V\in C^4, \quad V(0)=V'(0)=0, \quad V''(0)>0, \quad V'''(0)> 0.
\end{equation}
Well-known examples are given by the potential of the cubic FPU chain $V(r)=ar^2/2+b r^3/6$, the Lennard-Jones-(12,6) potential $V(r)=Ar^{-12}-Br^{-6}$, and the Toda potential $V(r)=\alpha (\E^{-\beta r}+\beta r-1)$. The latter gives rise to a completely integrable Hamiltonian system. 
Let us denote by
\begin{equation*}
r(j,t)=q(j+1,t)-q(j,t)
\end{equation*}
the distortion of the $j$th bond length out of equilibrium. Moreover, we introduce the shift operators $S^{\pm}q(j,t)=q(j\pm 1,t)$
for the backward and forward shifts along the lattice, respectively. With these notations the equations of motion \eref{equ:equofmot} become
\begin{equation}
\partial_t^2r(j,t)=(S^+-2+S^-)V'(r(j,t)).
\end{equation}
One can write this equation as a first-order Hamiltonian system of the form
\begin{equation} \label{equ:hamsystem}
\partial_t u=\left(\begin{array}{cc} 0 & S^+-1 \\ 1-S^- & 0\end{array}\right) DH(u),
\end{equation}
where $u=(r,p)$, $H(u)=\sum_{j\in \Z} \left( \frac{1}{2}p(j,t)^2+V(r(j,t))\right)$, and $DH(u)=(V'(r),p)$ denotes the functional gradient of $H$.
By making the travelling wave ansatz $r(j,t)=r_c(j-ct)$ equation \eref{equ:hamsystem} becomes
\begin{equation} \label{equ:equofmot2}
c^2r_c''(x)=(S^+-2+S^-)V'(r_c(x)).
\end{equation}

Since we aim to consider spatially periodic solutions, the function spaces appropriate for our setting will be the periodic Sobolev spaces

\begin{equation*}
H^1_{2L}(\R):= \Big\{ \phi\in \mathcal{P}'_{2L}: \| \phi\|_{H^1_{2L}}^2=2L\sum_{m=-\infty}^{+\infty}\big(1+(\frac{m\pi}{L})^2\big)|\what{\phi}(m)|^2<\infty\Big\},
\end{equation*}
where $\mathcal{P}'_{2L}$ is the set of $2L$-periodic distributions, i.e.~the set of all continuous linear functionals from $\mathcal{P}_{2L}=C^{\infty}_{2L}$ (the set of all smooth $2L$-periodic functions from $\R$ into $\C$) into $\C$.
The Fourier series of $r_c\in H^1_{2L}(\R)$ and its Fourier coefficients $\what{r}_c$ are given by
\begin{equation} \label{equ:Fouriertrafo}
r_c(x)=\sum_{m=-\infty}^{+\infty}\what{r}_c(m)\E^{\I\frac{m\pi}{L}x}, \qquad \what{r}_c(m)=\frac{1}{2L}\int_0^{2L}r_c(x)\E^{-\I\frac{m\pi}{L}x}dx.
\end{equation}
The space $H^1_{2L}(\R)$ is a Hilbert space with the inner product
\begin{equation*}
\langle \phi, \psi\rangle_{H^1_{2L}}=2L\sum_{m=-\infty}^{+\infty}\big(1+(\frac{m\pi}{L})^2\big)\what{\phi}(m)\ol{\what{\psi}(m)}.
\end{equation*}

\begin{remark} \label{rem:contlimes} Formally, as first noted by Zabusky and Kruskal \cite{zabusky1965interaction} the KdV equation arises from \eref{equ:equofmot2} as follows.
One assumes that there exist solutions satisfying the multiscale ansatz 
\begin{equation} \label{ZKansatz}
r_c(z)=\eps^2\Phi(\eps (x-ct)), \qquad \frac{c}{c_s}=1+\frac{\eps^2 c_{KdV}}{24},
\end{equation}
which describes waves travelling at approximately the speed of sound while showing a negligible temporal and spatial change in their shape. 
Then by Taylor expanding differences and neglecting terms of order $O(\eps^8)$ in equation \eref{equ:equofmot2} one obtains a KdV equation - arising as the coefficients of the $\eps^6$-terms - of the form
\begin{equation} \label{equ:KdVtransfo}
-\Phi''+6\frac{V'''(0)}{V''(0)}(\Phi^2)''+\Phi''''=0.
\end{equation}
From comparison of the $\eps^4$-terms one obtains the relation $c_s=\sqrt{V''(0)}$ for the speed of sound.
\end{remark}

The ansatz \eref{ZKansatz} can be justified rigorously via a renormalization appraoch adapted from \cite{friesecke1999solitary} (see Section~\ref{sec:renorm}), spectral theory for the Schr\"odinger operator with KdV
cnoidal wave potential (see Section~\ref{sec:KdVsolutions} and~\ref{app:Lame}), and an implicit function theorem argument (see Section~\ref{sec:existence}). This leads to the following rigorous persistence result for 
cnoidal waves
when deforming the limiting KdV equation back into the lattice equation. 

\begin{theorem} \label{thm:mainresult}
Assume that the nearest neighbour interaction potential $V$ satisfies the assumptions \eref{equ:hypopotential}, 
and let $c_s=\sqrt{V''(0)}$ denote the sonic wave velocity. Fix the constants $k_0\in(0,1)$ and 
$L_0\in\R^+$ such that $c_{KdV}(k_0,L_0)=4K^2(k_0)\sqrt{1-k_0^2+k_0^4}/L_0^2=1$, where $K(k)$ denotes the
complete elliptic integral of the first kind, i.e.~$K(k)=\int_0^{\pi/2}ds/\sqrt{1-k^2\sin s}$.
Then the following statements hold:
\begin{itemize}
\item[(a)] (existence and local uniqueness) Given $\delta>0$ sufficiently small, there exists $\eps_0>0$ such that, for $\eps \in (0,\eps_0)$ and $c^2/c_s^2=(1+\eps^2/12)$, the governing equation \eref{equ:equofmot2} for the periodic travelling wave profile admits a unique solution $r_c$ in the set
\begin{equation*}
\{ r\in H^1_{2L_0}(\R)| r\textrm{ even, } \vert| \eps^{-2}r(\eps^{-1}. )-\Phi_1^{(k_0,L_0)}\vert|_{H^1_{2L_0}}<\delta\},
\end{equation*}
where $\Phi_1^{(k_0,L_0)}$ denotes the KdV cnoidal wave profile with speed $c_{KdV}(k_0,L_0)=1$ that solves the integrated KdV travelling wave equation
\begin{equation} \label{equ:assocKdV}
-\Phi+\Phi''+6\frac{V'''(0)}{V''(0)}\Phi^2=0.
\end{equation}
Explicitly, $\Phi_1^{(k_0,L_0)}$ is given by
\begin{equation}
\Phi_1^{(k_0,L_0)}(\xi)=\frac{V''(0)}{V'''(0)}\frac{K^2(k_0)}{L_0^2}\left( \frac{1-2k_0^2+\sqrt{1-k_0^2+k_0^4}}{3}+k_0^2\cn^2\left(\frac{K(k_0)}{L_0}\xi ;k_0^2 \right)\right),
\end{equation}
Here $\cn$ denotes one of the Jacobian elliptic functions (see Section~\ref{sec:KdVsolutions} for its precise definition).
\item[(b)] (asymptotic shape) The solution $r_c$ from (a) satisfies the estimate
\begin{equation} \label{shape}
\left\| \frac{1}{\eps^2}r_c\left( \frac{.}{\eps} \right)-\Phi_1^{(k_0,L_0)}\right\|_{H^1_{2L_0}}\leq C\eps^2,
\end{equation}
where $C$ is independent of $\eps$.
\item[(c)] (smoothness) The mapping $c\mapsto r_c$ from $(c_s,\infty)$ into $H_{2L_0}^1$ is $C^1$.
\end{itemize}
\end{theorem}
Estimate (\ref{shape}) means that the wave profile $r_c$ that solves \eref{equ:equofmot2} has a characteristic period of order $1/\eps$ and amplitude of order $\eps^2$.

\section{Renormalization and the continuum limit} \label{sec:renorm}

The formal argument in Remark~\ref{rem:contlimes} which relates the lattice equation to the KdV equation is not mathematically satisfactory. This is because it involves
uncontrolled truncation of a Taylor expansion of a difference operator into (more and more unbounded) difference operators. To overcome this problem, Friesecke and Pego
\cite{friesecke1999solitary} introduced -- in the context of spatially localized waves -- a renormalization approach which recasts both the lattice travelling wave equation \eref{equ:equofmot2} and the limiting
KdV travelling wave equation in a form involving only {\it bounded} operators (see \eref{equ:fixedpointr} below). The ensuing lattice Fourier multiplication operator governing lattice waves and its limiting continuum counterpart can then be shown, via a careful analysis of the location of their complex
poles, to be rigorously close in an appropriate operator norm. 
As we show in this section, this approach can be adapted to the periodic setting.
However, as we will see in Section~\ref{sec:KdVsolutions}, the limiting KdV equation (\ref{equ:assocKdV}) becomes more subtle in the periodic setting, admitting a two-parameter family of
solutions in place of the one-parameter soliton family.

We now derive the renormalized form of (\ref{equ:equofmot2}) and its small-amplitude long-wave limit, by adapting the analogous steps in \cite{friesecke1999solitary} to the periodic case.

Let us make the assumptions \eref{equ:hypopotential} on the nearest-neighbour potential $V$. We isolate the nonlinear from the linear part of the restoring force by writing
\begin{equation*}
V'(r)=V''(0)r+N(r), \qquad N(r)=\frac{1}{2}V'''(0)r^2\big(1+\eta (r)\big),
\end{equation*}
where $\eta(r)\in C^1$ with $\eta(0)=0$. 
Let us consider functions $r_c$ in the periodic Sobolev space $H^1_{2L}(\R)$ and their corresponding Fourier transform \eref{equ:Fouriertrafo}. Then \eref{equ:equofmot2} turns into the equation
\begin{equation*}
\Big( c^2\big( \frac{m\pi}{L}\big)^2-4V''(0)\sin^2\big( \frac{m\pi}{2L}\big) \Big)\what{r}_c(m)=4\sin ^2\big( \frac{m\pi}{2L}\big)\what{N(r_c)}(m).
\end{equation*}
Since $c^2\big( \frac{m\pi}{L}\big)^2-4V''(0)\sin^2\big( \frac{m\pi}{2L}\big)>0$ for $c^2>1$ and $m\in\Z\setminus \{0\}$, this equation can be transformed into the fixed point equation
\begin{equation} \label{equ:fixedpointr}
r_c=P \ N(r_c),
\end{equation}
where $P$ denotes the pseudo-differential operator $\widehat{Pr}(m)=p(m)\hat{r}(m)$ with symbol
\begin{equation} \label{equ:pseudosymbol}
p(m)=\frac{4\sin^2\big( \frac{m\pi}{2L}\big)}{c^2\big(\frac{m\pi}{L}\big)^2-4V''(0)\sin^2\big( \frac{m\pi}{2L}\big)}=\frac{\sinc^2\big(\frac{m\pi}{2L}\big)}{c^2-V''(0)\sinc^2\big(\frac{m\pi}{2L}\big)}.
\end{equation}
Here $\sinc(x)=\sin(x)/x$ denotes the usual sinc function.

Let us introduce a small parameter $\eps>0$ by setting
\begin{equation} \label{equ:scalingc}
c=\sqrt{V''(0)}\left( 1+\frac{\eps^2c_{KdV}}{24}\right), \quad \textrm{with } c_{KdV}=\frac{4K^2(k)}{L^2}\sqrt{1-k^2+k^4},
\end{equation}
where $k^2\in (0,1)$ and $L\in\R^+$ are fixed real constants.
By rescaling variables via
\begin{equation} \label{equ:renorm}
\Phi^{(\eps,k,L)}(x)=\frac{1}{\eps^2}r_c\big( \frac{x}{\eps}\big)
\end{equation}
the renormalization equation \eref{equ:fixedpointr} turns into an ($\eps$-dependent) fixed point equation for $\Phi^{(\eps,k,L)}$:
\begin{equation} \label{equ:fixedpointphieps}
\Phi=P^{(\eps)}N^{(\eps)}(\Phi),
\end{equation}
where $N^{(\eps)}(\Phi)=\frac{1}{2}V'''(0)\Phi^2\big( 1+\eta (\eps^2\Phi)\big)$ and the operator $P^{(\eps)}$ has the symbol
\begin{equation} \label{equ:symbolpeps}
p^{(\eps)}(m)=\frac{\eps^2\sinc^2\big(\eps \frac{m\pi}{2L}\big)}{c^2-V''(0)\sinc^2\big(\eps \frac{m\pi}{2L}\big)}.
\end{equation}

Next, we recall from \cite{friesecke1999solitary} how the system \eref{equ:fixedpointphieps}, (\ref{equ:symbolpeps}) formally reduces to KdV as $\eps\to 0$.
In the small $\eps$ regime the wave speed scaling \eref{equ:scalingc} implies
\begin{equation} \label{equ:scalingcnew}
c^2=V''(0)\left( 1+\frac{\eps^2c_{KdV}}{12}\right) + O(\eps^4).
\end{equation}
From this and the Taylor expansion $\sinc^2(x)=1-x^2/3+O(x^4)$ one obtains for the pointwise limit of of the symbol as $\eps\to 0$:
\begin{equation} \label{equ:symbolp0}
\lim_{\eps\to 0}p^{(\eps)}(m)=p^{(0)}(m)=\frac{12}{V''(0)}\frac{1}{ c_{KdV}+\big(\frac{m\pi}{L}\big)^2}.
\end{equation}
The nonlinearity satisfies $\lim_{\eps\to 0}N^{(\eps)}(\Phi)=\frac12 V'''(0)\Phi^2$.
Hence as $\eps\to 0$ the fixed point equation \eref{equ:fixedpointphieps} converges to the equation 
\begin{equation} \label{equ:fixedpointphi}
\Phi=P^{(0)}N^{(0)}(\Phi)=6\frac{V'''(0)}{V''(0)}(c_{KdV}-\partial^2)^{-1}\Phi^2.
\end{equation}
By applying the operator $(c_{KdV}-\partial^2)$ to both sides, we see that this equation is equivalent to the integrated KdV equation
\begin{equation} \label{equ:intKdVPhi}
-c_{KdV}\Phi+6\frac{V'''(0)}{V''(0)}\Phi^2+\Phi''=0,
\end{equation}
which stems from integrating the usual KdV travelling wave equation 
\begin{equation}\label{equ:nonintKdVPhi}
-c_{KdV}\Phi'+12\frac{V'''(0)}{V''(0)}\Phi \Phi'+\Phi'''=0, 
\end{equation}
with vanishing integration constant. 

A key advantage of the above analysis over the formal asymptotics in Section~\ref{sec:result} is the
fact, shown in \cite{friesecke1999solitary}, that the convergence in (\ref{equ:symbolp0}) is in fact uniform.
This will allow us to show that the
multipliers $P^{(\eps)}$ and $P^{(0)}$ are close in operator norm on the periodic Sobolev space $H^1_{2L}(\R)$  introduced in the previous section.
It is useful to introduce the notation
\begin{equation} \label{equ:tildenotation}
  p^{(\eps)}(m)=:\tilde{p}^{(\eps)}(\frac{m\pi}{L}), \;\; p^{(0)}(m)=:\tilde{p}^{(0)}(\frac{m\pi}{L});
\end{equation}
the $L$-independent functions $\tilde{p}^{(\eps)}$ and $\tilde{p}^{(0)}$ coincide with those considered in \cite{friesecke1999solitary}. 

\begin{lemma} (Fourier multiplier estimate) \label{lem:pseudoestimate}
Fix a point $(k_0,L_0)\in (0,1)\times \R^+$ such that $c_{KdV}(k_0,L_0)= 1$. Then there exists a constant $C>0$ such that for all 
sufficiently small $\eps>0$ and $c^2=V''(0)(1+\eps^2/12)$ the Fourier multipliers defined in \eref{equ:symbolpeps}, \eref{equ:symbolp0}, 
\eref{equ:tildenotation} satisfy
$$
a) \;\;\;\;\; \sup_{s\in\R} \vert\tilde{p}^{(\eps)}(s)-\tilde{p}^{(0)}(s)\vert\leq C \eps^2.
$$
$$
b) \;\;\;\;\; \| P^{(\eps)}-P^{(0)}\| _{\mathcal{L}(H_{2L_0}^1)} \leq C\eps^2,
$$
where $\mathcal{L}(X)$ denotes the Banach space of bounded linear operators from the Banach space $X$ into itself.
\end{lemma}

\begin{proof}
a) was proved in \cite[Lemma~3.1]{friesecke1999solitary}, by careful analysis of the location of the poles of $\tilde{p}^{(\eps)}$ in the complex plane.
To show b), note that, for any operator on $H^1_{2L}(\R)$ of form $\widehat{Ar}(m)=\tilde{a}(m\pi/L)\hat{r}(m)$, 
$$
   \| A \|^2_{L(H^1_{2L}} = \sup_{r\in H^1_{2L}\backslash\{0\} } \frac{ 2 L \sum_{m\in\Z} (1+(m\pi/L)^2)|\tilde{a}(m\pi/L)\hat{r}(m)|^2}
                                       { 2 L \sum_{m\in\Z} (1+(m\pi/L)^2)|\hat{r}(m)|^2} 
                          = \sup_{s\in\pi\Z/L} |\tilde{a}(s)|^2.
$$
Applying this with $A=P^{(\eps)}-P^{(0)}$, $\tilde{a}=\tilde{p}^{(\eps)}-\tilde{p}^{(0)}$ yields 
$$
    \| P^{(\eps)}-P^{(0)}\| _{\mathcal{L}(H_{2L_0}^1)} = \sup_{s\in\pi\Z/L} |\tilde{p}^{(\eps)} - \tilde{p}^{(0)}|. 
$$
Estimating the above supremum by that over $s\in\R$ and applying a) yields the assertion. 
\end{proof}

\section{Periodic KdV travelling wave solutions} \label{sec:KdVsolutions}
In this section we explicitly describe the periodic travelling wave solutions of the (integrated) KdV equation \eref{equ:intKdVPhi}, see also \cite{drazin1989solitons,pava2006stability}. Here a main difference from the solitary wave case is present. Mathematically, there arises 
a non-zero integration constant after multiplying the KdV equation \eref{equ:assocKdV} by $\Phi'$ and integrating. This forces us to deal with a 2-parameter system (for instance represented physically by the speed and period of the wavetrain) in contrast to the soliton case, which is fully described by one parameter corresponding to the velocity (or equivalently the amplitude or width) of the single soliton. 
The class of periodic solutions contains soliton solutions as a degenerate limit, namely infinite period.

By multiplying the (integrated) KdV equation \eref{equ:intKdVPhi} with $\Phi'$ and integrating once more, we get
\begin{equation} \label{equ:intKdVtwice}
(\Phi')^2=4\frac{V'''(0)}{V''(0)}\left( -\Phi^3+\frac{V''(0)}{4V'''(0)}c_{KdV} \Phi^2+\frac{V''(0)}{2V'''(0)}B_{\Phi}\right).
\end{equation}
Let us denote the polynomial in the bracket on the right hand side by $F(\Phi)$. For finding periodic solutions of \eref{equ:intKdVtwice} it is useful to consider the roots of $F(\Phi)$. The function can by factorized by
\begin{equation*}
F(\Phi)=-(\Phi-E_1)(\Phi-E_2)(\Phi-E_3),
\end{equation*}
where $E_1$, $E_2$, and $E_3$ denoted the three roots.
It turns out that real periodic solutions occur if the three zeros of the polynomial $F(\Phi)$ are real and distinct, i.e.~$E_1<E_2<E_3$.
The dependence of the zeros on the parameters $c_{KdV}$ and $B_{\Phi}$ is thus given by
\begin{eqnarray} 
&&E_1+E_2+E_3=\frac{V''(0)}{4V'''(0)}c_{KdV}, \label{equ:system1}\\
&&E_1E_2+E_2E_3+E_1E_3=0,\label{equ:system2}\\
&&E_1E_2E_3=\frac{V''(0)}{2V'''(0)}B_{\Phi}.\label{equ:system3}
\end{eqnarray}
The real and bounded solutions are contained in the region $[E_2,E_3]$, whereas in the region $(-\infty,E_1]$ the solutions are unbounded. Introduce the normalized variable $\rho=\Phi/E_3$, then \eref{equ:intKdVtwice} becomes
\begin{equation} \label{equ:intKdVtwicerho}
(\rho')^2=-4\frac{V'''(0)}{V''(0)}E_3(\rho-\rho_1)(\rho-\rho_2)(\rho-1),
\end{equation}
where we have set $\rho_i=E_i/E_3$ for $i=2,3$. The new variable $\rho$ takes values in the interval $[\rho_2,1]$. Assume that a maximum value of $\rho$ is at $z=0$ which can always be achieved by a coordinate translation. Since the critical points of $\rho\in [\rho_2,1]$ are at the boundary of the interval we conclude $\rho(0)=1$. Moreover, we have $\rho''(\rho_2)>0$ and $\rho''(1)<0$. Hence a bounded solution $\rho$ oscillates between the values $\rho_2$ and $1$.

In order to construct such a solution let us again make a change of variables by defining $\psi$ implicitly via
\begin{equation*}
\rho=1+(\rho_2-1)\sin^2\psi
\end{equation*} 
where $\psi$ is continuous with $\psi(0)=0$. With this definition \eref{equ:intKdVtwicerho} becomes 
\begin{equation} \label{equ:intKdVtwicepsi}
(\psi')^2=\frac{V'''(0)}{V''(0)}E_3(1-\rho_1)\big( 1-\frac{1-\rho_2}{1-\rho_1}\sin^2\psi\big).
\end{equation}
We want to use the standard notation in this context and therefore define the parameters
\begin{equation*}
k^2=\frac{1-\rho_2}{1-\rho_1}, \qquad \lambda=\frac{V'''(0)}{V''(0)}E_3(1-\rho_1),
\end{equation*}
which fulfill $k^2\in (0,1)$ and $\lambda>0$.
Using this notation \eref{equ:intKdVtwicepsi} can be written in the form
\begin{equation*}
(\psi')^2=\lambda (1-k^2\sin^2\psi),
\end{equation*}
which can be solved implicitly by integration
\begin{equation} \label{equ:solellint}
I(\psi;k^2):=\int_0^{\psi(\xi)}\frac{ds}{\sqrt{1-k^2\sin^2s}}=\sqrt{\lambda}\xi+\xi_0,
\end{equation}
where $\xi_0$ is a final arbitrary constant of integration. Here $I(\psi;k^2)$ is called the standard elliptic integral of the first kind and $k^2$ the elliptic modulus. The Jacobian elliptic function $\sn$ is defined as the inverse of the function $\psi \mapsto I(\psi;k^2)$ via $\sin\psi=\sn(\sqrt{\lambda}\xi+\xi_0;k^2)$. Another basic elliptic function, the cnoidal function $\cn$, is defined in terms of $\sn$ via $\cn(z ;k^2)=\sqrt{1-\sn^2(z ;k^2)}$. Thus
\begin{equation*}
\rho=1+(\rho_2-1)\sn^2(\sqrt{\lambda}\xi;k^2)=\rho_2+(1-\rho_2)\cn^2(\sqrt{\lambda}\xi+\xi_0;k^2).
\end{equation*}
Hence by recasting variables appropriately the solution of \eref{equ:intKdVPhi} is finally given by the so-called cnoidal wave
\begin{equation} \label{equ:cnoidalsol}
\Phi(\xi)=E_2+(E_3-E_2)\cn^2 \Big( \sqrt{\frac{V'''(0)}{V''(0)}(E_3-E_1)}\xi+\xi_0; k^2\Big),
\end{equation}
where the travelling wave coordinate is
\begin{equation*} 
\xi=x-\frac{4V'''(0)}{V''(0)}(E_1+E_2+E_3)t,
\end{equation*}
and the elliptic modulus is given by
\begin{equation*}
k^2=\frac{E_3-E_2}{E_3-E_1}.
\end{equation*}
The solution \eref{equ:cnoidalsol} is periodic if $k^2\in (0,1)$ with period 
\begin{equation*}
2L=\frac{2\sqrt{V''(0)}K(k)}{\sqrt{V'''(0)(E_3-E_1)}},
\end{equation*}
where 
\begin{equation*}
K(k):=I(\pi/2;k^2)=\int_0^{\pi/2}\frac{ds}{\sqrt{1-k^2\sin^2s}}
\end{equation*}
is called the complete elliptic integral of the first kind.

It is worthwhile to consider two limiting cases of these KdV cnoidal wave solutions.
\begin{itemize}
\item[(i)] \emph{The "most nonlinear" limit}:  In the case $E_2\to E_1$ (corresponding to $k^2\to 1$) we have $K(k)\to +\infty$ and $\cn(z ;k^2)\to\sech z$. Thus in this case the solution of \eref{equ:KdVtransfo} with $\xi_0=0$ corresponds to a single soliton solution 
\begin{equation*}
\Phi(\xi)=E_1+(E_3-E_1)\sech^2\Big( \sqrt{\frac{V'''(0)}{V''(0)}(E_3-E_1)}(x-c_{KdV}t)\Big),
\end{equation*}
where $c_{KdV}=4V'''(0)(2E_1+E_3)/V''(0)$. Due to the invariance of the KdV equation under Galilean transformations, i.e.~transformations of the form
\begin{equation*}
\Phi(\xi)\longrightarrow a+\Phi \Big(\xi-12a\frac{V'''(0)}{V''(0)}t\Big),
\end{equation*}
this solution is equivalent to the well-known standard form
\begin{equation} \label{equ:onesoliton}
\Phi(\xi)=\frac{V''(0)}{V'''(0)}\Big( \frac{\sqrt{\beta}}{2}\sech \big( \frac{\sqrt{\beta}}{2}(x-\beta t)\big)\Big)^2, 
\end{equation}
where $\beta=4V'''(0)(E_3-E_1)/V''(0)$.
\item[(ii)] \emph{The "linear" limit}: In the case $E_3\to E_2$ which corresponds to the limit $k^2\to 0$ (amplitude tending to zero) we have $K(k)\to \pi/2$ and $\cn(z;k^2)\to\cos z$. Therefore the limiting behaviour of \eref{equ:KdVtransfo} with $\xi_0=0$ as $E_3\to E_2$ is
\begin{equation} \label{equ:linearlimit}
\Phi(\xi)\approx \frac{E_3+E_2}{2}+\frac{(E_3-E_2)}{2}\cos\Big( \sqrt{\frac{4V'''(0)}{V''(0)}(E_3-E_1)}(x-c_{KdV}t)\Big),
\end{equation}
where $c_{KdV}=4V'''(0)(E_1+2E_2)/V''(0)$. Here we made use of $\cos^2(x)=(\cos(2x)+1)/2$. Thus in leading orders of this small amplitude limit the cnoidal wave describes a linear wave. 
\end{itemize}

From the considerations above we deduce that the periodic solutions of the type~\eref{equ:cnoidalsol} can be considered an intermediate form between linear waves and highly nonlinear solitons. 

By solving the KdV travelling wave equation \eref{equ:nonintKdVPhi} via integration one has in general to deal with a 3-parameter system, stemming from the cubic polynomial $F(\Phi)$. Thus the KdV solutions constructed in this way depend on 3 parameters, reflected by the roots $\{E_1,E_2,E_3\}$, the integration constants and velocity $\{c_{KdV},A_{\Phi},B_{\Phi}\}$, or elliptic modulus, period and velocity $\{k^2,2L,c_{KdV}\}$. Due to the choice of the renormalization ansatz \eref{equ:renorm} the governing equations \eref{equ:equofmot2} for the string are related to an integrated KdV equation \eref{equ:intKdVPhi} with zero integration constant, i.e.~$A_{\Phi}=0$. Hence the 3-parameter reduces to a 2-parameter system. We aim to admit periodic KdV travelling waves to our considerations; for that purpose it is essential that the second integration constant $B_{\Phi}$ is non-zero. Otherwise the system would include only one free parameter giving rise to soliton solutions.

The cnoidal wave solution \eref{equ:cnoidalsol} is given in terms of the roots $E_1$, $E_2$, and $E_3=E_3(E_1,E_2)$. In the following we derive an equivalent expression in the parameters $k$, $L$, and $c_{KdV}=c_{KdV}(k,L)$. Since the latter model directly the "physical" quantities like the shape, period, and velocity of the wavetrain this will help us to get a less abstract picture of the solutions we are dealing with. Besides it enables us to directly relate our spectral analysis results derived in Section~\ref{subsec:specanalysis} to studies on the Lam\'e equation, cf.~\ref{app:Lame}.

Let us start by making the ansatz $\Phi(\xi)=A+B\cn^2(D\xi;k^2)$ with $\xi=x-Ct$. Insert this expression into \eref{equ:intKdVPhi} and use the relations $\cn'x=-\cn x \ \sn x \ \dn x$, $\sn^2x+\cn^2x=1$, and $\dn^2x=1-k^2+k^2\cn^2x$. Then comparing coefficients yields the following system of equations
\begin{eqnarray}
6\frac{V'''(0)}{V''(0)}A^2-CA+2(1-k^2)BD^2=0,&& \label{equ:eqsys1}\\
C-12\frac{V'''(0)}{V''(0)}A+4(1-2k^2)D^2=0,&& \label{equ:eqsys2}\\
B-\frac{V''(0)}{V'''(0)}k^2D^2=0.&& \label{equ:eqsys3}
\end{eqnarray}
Since we assume $\Phi$ to be periodic with period $2L$ and $\cn^2$ is a $2K(k)$-periodic function we have $D=K(k)/L$. Hence from \eref{equ:eqsys3} we get
\begin{equation*}
B=\frac{V''(0)}{V'''(0)}\frac{K^2(k)k^2}{L^2}.
\end{equation*}
Moreover, from \eref{equ:eqsys2} one obtains
\begin{equation*}
C=12\frac{V'''(0)}{V''(0)}A-4(1-2k^2)\frac{K^2(k)}{L^2}.
\end{equation*}
Finally, inserting these formulas into \eref{equ:eqsys3} gives the following quadratic equation
\begin{equation*}
A^2-2(1-2k^2)\frac{V''(0)}{V'''(0)}\frac{K^2(k)}{3L^2}A-(k^2-k^4)\left(\frac{V''(0)}{V'''(0)}\right)^2\frac{K^4(k)}{3L^4}=0,
\end{equation*}
which has the positive solution
\begin{equation*}
A=\frac{V''(0)}{V'''(0)}\frac{K^2(k)}{3L^2}\left( 1-2k^2+\sqrt{1-k^2+k^4}\right).
\end{equation*}
Thus we can write the cnoidal wave solution \eref{equ:cnoidalsol} in the convenient form
\begin{equation} \label{equ:cnoidalsolkL}
\Phi(\xi)=\frac{V''(0)}{V'''(0)}\frac{K^2(k)}{L^2}\left( \frac{1-2k^2+\sqrt{1-k^2+k^4}}{3}+k^2\cn^2 \Big( \frac{K(k)}{L}\xi;k^2\Big)\right),
\end{equation}
where $\xi=x-c_{KdV}t$ is the travelling wave coordinate with the velocity
\begin{equation*}
c_{KdV}=\frac{4K^2(k)}{L^2}\sqrt{1-k^2+k^4}.
\end{equation*}
Note that for $k^2\in(0,1)$ we have $c_{KdV}\in (\pi^2/L^2,\infty)$. The representation \eref{equ:cnoidalsolkL} of the cnoidal wave solution is indeed equivalent to \eref{equ:cnoidalsol} if we choose 
\begin{eqnarray*}
&&E_1=\frac{V''(0)}{V'''(0)}\frac{K^2(k)}{3L^2}\big( -2+k^2+\sqrt{1-k^2+k^4}\big),\\
&&E_2=\frac{V''(0)}{V'''(0)}\frac{K^2(k)}{3L^2}\big( 1-2k^2+\sqrt{1-k^2+k^4}\big),\\
&&E_3=\frac{V''(0)}{V'''(0)}\frac{K^2(k)}{3L^2}\big( 1+k^2+\sqrt{1-k^2+k^4}\big).
\end{eqnarray*}

In Figure~\ref{fig:cnoidalLconst} cnoidal wave solutions of the form~\eref{equ:cnoidalsolkL} for different parameters $k^2$ and $c_{KdV}$ are shown, that is they are normalized to have the same period.

\begin{figure} 
\centering
\includegraphics[width=.8\textwidth]{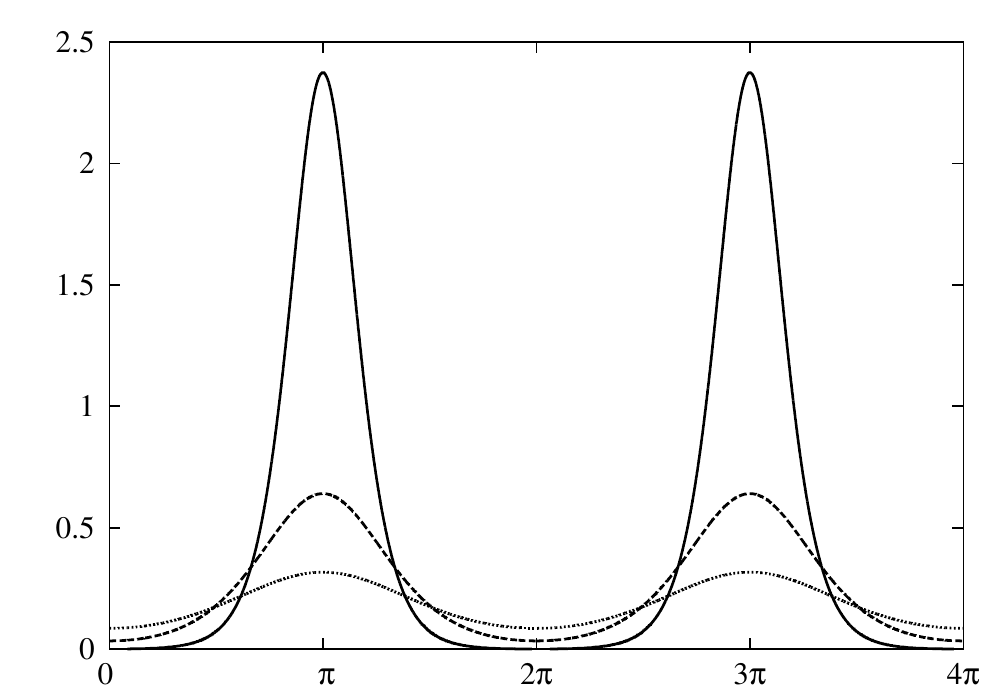}
\caption{Cnoidal wave solutions~\eref{equ:cnoidalsolkL} of \eref{equ:intKdVPhi} normalized to have the same period. $2L=2\pi$. We have set $V''(0)=V'''(0)=1$. The three curves represent the solutions corresponding to the different values for the elliptic modulus $k^2=0.6$ (dotted), $k^2=0.9$ (dashed), and $k^2=0.999$ (solid).\label{fig:cnoidalLconst}}
\end{figure}

\begin{remark}
Note that an equivalent renormalization approach would be the following: replace the renormalization ansatz \eref{equ:renorm} by $\phi^{(\eps,k,L)}(x)+E_2=\eps^{-2}r_c(\eps^{-1}x)$, where $E_2\neq 0$ denotes an arbitrary constant. In the limit $\eps\to 0$ this leads to an integrated KdV equation for $\phi$ given by 
\begin{equation*}
-\tilde{c}_{KdV}\phi+6\frac{V'''(0)}{V''(0)}\phi^2+\phi''+A_{\phi}=0,
\end{equation*}
where $A_{\phi}\neq 0$. Indeed, just set $\Phi=\phi+E_2$ in \eref{equ:intKdVPhi}. When integrating this equation once more assume the second integration constant $B_{\phi}$ to vanish. Then one of the zeros $e_1<e_2<e_3$ of the corresponding cubic polynomial has to vanish, say $e_2=0$. The cnoidal wave solution is given by
\begin{equation*}
\phi(\xi)=e_3\cn^2\left( \sqrt{V''(0)(e_3-e_1)/V'''(0)}\xi;k^2\right), \qquad \textrm{with }\xi=x-\tilde{c}_{KdV}t,
\end{equation*} 
where the velocity is $\tilde{c}_{KdV}=4V'''(0)(e_1+e_3)/V''(0)$. If we now add the constant $E_2$, that is if we make the transformation $\phi \to \phi +E_2$, one has to shift the velocity simultaneously, $\tilde{c}_{KdV} \to \tilde{c}_{KdV}+12V'''(0)E_2/V''(0)$, to obtain an equivalent solution (Galilean invariance). Indeed, this yields the solution $\Phi$ given by \eref{equ:cnoidalsol} with $E_i=e_i+E_2$ for $i=1,2,3$.
\end{remark}

\section{Existence and shape near the continuum limit} \label{sec:existence}

In this section we prove our main result, Theorem~\ref{thm:mainresult}. The technical ingredients of the proof are
\begin{itemize}
\item[(i)] the renormalized form of the lattice travelling wave equation and associated lattice Fourier multiplier estimate given in Section~\ref{sec:renorm},
\item[(ii)] spectral information about the linearization of \eref{equ:fixedpointr} at the KdV cnoidal wave (results from direct scattering theory for Schr\"odinger operators with KdV cnoidal wave-potential), and
\item[(iii)] a quantified version of the standard implicit function theorem borrowed from \cite{friesecke1999solitary}.
\end{itemize}
We start by fixing the wave speed $c_{KdV}$ of the KdV travelling wave profile by fixing the two parameters $k$ and $L$ of our system. We do so by choosing a point $k_0\in(0,1)$ and setting $L_0^2=4K^2(k_0)\sqrt{1-k_0^2+k_0^4}$ such that $c_{KdV}(k_0,L_0)= 1$.
That means, we fix the parameters $k$ and $L$ in \eref{equ:renorm} and \eref{equ:scalingcnew}. 

In Figure~\ref{fig:cnoidalcconst} cnoidal wave solutions of the form~\eref{equ:cnoidalsolkL} for different parameters $k^2$ are shown. By setting $L^2=4K^2(k)\sqrt{1-k^2+k^4}$ they are normalized to have the velocity $c_{KdV}=1$.

Next, we introduce
\begin{equation*}
F(\Phi)=P^{(0)}N^{(0)}(\Phi), \qquad \tilde{F}(\Phi)=P^{(\eps)}N^{(\eps)}(\Phi)-P^{(0)}N^{(0)}(\Phi).
\end{equation*}
The map $F$ is smooth on $H^1_{2L_0}$ and it has a fixed point $\Phi^{(k_0,L_0)}_1$ given by \eref{equ:cnoidalsolkL}. The derivative of $F$ at $\Phi^{(k_0,L_0)}_1$ is the operator $L=DF(\Phi^{(k_0,L_0)}_1)$ defined by
\begin{equation} \label{equ:linearizationop}
L\psi=12\frac{V'''(0)}{V''(0)}\frac{\Phi^{(k_0,L_0)}_1}{c_{KdV}-\partial^2}\psi.
\end{equation}
We will solve the equation $\Phi-F(\Phi)-\tilde{F}(\Phi)=0$ by the inverse function theorem, using that (i) $\tilde{F}$ is small in a suitable sense for small $\eps$ (using the estimate of Lemma~\ref{lem:pseudoestimate}) and that (ii) the linearization of the map $\Phi-F(\Phi)-\tilde{F}(\Phi)$ at the fixed point $\Phi^{(k_0,L_0)}_1$ for $\eps=0$ is $I-L$, which has a bounded inverse on the space of even $2L_0$-periodic functions $E_{2L_0}^1=\{ \Phi\in H_{2L_0}^1 | \Phi \textrm{ even}\}$, cf.~Lemma~\ref{lem:specanalysis}.

\subsection{Spectral analysis} \label{subsec:specanalysis}

To apply the implicit function theorem we have to establish spectral properties of the operator $F(\Phi)=P^{(0)}N^{(0)}(\Phi)$ linearized at a KdV cnoidal wave solution $\Phi^{(k,L)}_1$. For our purpose it suffices to show that the eigenvalue $\lambda=1$ is simple with corresponding odd $2L$-periodic eigenfunction $\frac{d}{d\xi}\Phi^{(k,L)}_1$, such that $I-DF(\Phi^{(k,L)}_1)$ has a bounded inverse on the space of even $2L$-periodic functions.

It turns out that the spectral problem for $L=DF(\Phi^{(k,L)}_1)$ is connected with the scattering theory for Schr\"odinger operators with cnoidal wave potential. There exists a large literature on this eigenvalue equation, known as the Lam\'e equation. Therefore, an alternative way to obtain the below result via the theory of the Lam\'e equation is given in~\ref{app:Lame}.

\begin{lemma} \label{lem:specanalysis}
Fix a point $(k_0,L_0)\in (0,1)\times \R^+$ such that $c_{KdV}(k_0,L_0)= 1$. Then the operator $L=DF(\Phi^{(k_0,L_0)}_1)$ given by \eref{equ:linearizationop} is a compact self-adjoint operator on $H^1_{2L_0}(\R)$. The eigenvalues $\lambda_1=1$ and $\lambda_2=2$ are simple with corresponding eigenfunctions $\psi_1(\xi)=\frac{d}{d\xi}\Phi^{(k_0,L_0)}_1(\xi)$ and $\psi_2(\xi)=\Phi^{(k_0,L_0)}_1(\xi)$, respectively. Hence the operator $I-L$ has a bounded inverse on the space of even $2L_0$-periodic functions $E^1_{2L_0}=\{ \phi\in H_{2L_0}^1 | \phi \textrm{ even}\}$.
\end{lemma}
\begin{proof}
That $L$ is a self-adjoint operator follows from
\begin{eqnarray*}
&&\langle L\psi,\phi\rangle_{H^1_{2L_0}}=2L_0\sum_{m=-\infty}^{+\infty}(1+(m\pi/L_0)^2)\what{L\psi}(m)\ol{\what{\phi}(m)}\\
&&=2L_0\frac{12V'''(0)}{V''(0)}\sum_{m=-\infty}^{+\infty}\left(\Phi_1^{(k_0,L_0)}\psi\right)^{\what{}}(m)\ol{\what{\phi}(m)}=\frac{12V'''(0)}{V''(0)}\int_{0}^{2L_0}(\Phi_1^{(k_0,L_0)}\psi \phi)(x)dx .
\end{eqnarray*}

The compactness of $L$ is a consequence of the fact that $L$ is a bounded operator from $H^1_{2L_0}$ into $H^2_{2L_0}$ and the Rellich compactness theorem, which guarantees that $H^2_{2L_0}$ is compactly embedded into $H^1_{2L_0}$.

Concerning the spectral properties of $L$ we observe the following: since $L$ is a compact self-adjoint operator all eigenvalues are real and semi-simple. Indeed, all eigenvalues are simple: note that the spectral problem $L\psi=\lambda\psi$ where $L$ is given by \eref{equ:linearizationop} is equivalent to the second order differential equation
\begin{equation} \label{equ:equL}
-c_{KdV}\psi-\frac{12}{\lambda}\frac{V'''(0)}{V''(0)}\Phi_1^{(k_0,L_0)}\psi-\psi''=0,
\end{equation}
which can not have two independent bounded solutions. Moreover, if $\lambda=1$ then a solution of \eref{equ:equL} is $\psi=\frac{d}{d\xi}\Phi_1^{(k_0,L_0)}$ since $\Phi_1^{(k_0,L_0)}$ is a solution of the KdV equation \eref{equ:nonintKdVPhi}. If $\lambda=2$ then a solution of \eref{equ:equL} is $\psi=\Phi_1^{(k_0,L_0)}$ since $\Phi_1^{(k_0,L_0)}$ is a solution of the integrated KdV equation \eref{equ:intKdVPhi}.

\end{proof}

\subsection{A quantified version of the standard implicit function theorem}

The following result will be borrowed from \cite{friesecke1999solitary}, where it has been applied to the solitary wave case.

\begin{lemma}[Quantitative version of the implicit function theorem, \cite{friesecke1999solitary} Lemma~A.1] \label{lem:IFT}
Suppose $F$ and $\tilde{F}$ are $C^1$ maps from a ball in a Banach space $E$ into $E$. Suppose $F$ has a fixed point $\Phi_1$, i.e.~$F(\Phi_1)=\Phi_1$. Let $L=DF(\Phi_1)$ denote the linearization of $F$ at $\Phi_1$. Then assume that the operator $I-L$ is invertible with $|(I-L)^{-1}|\leq C_0<\infty$. Assume that $\delta>0$ is sufficiently small such that
\begin{equation*}
|DF(\Phi)-DF(\Phi_1)|\leq C_1<C_0^{-1} \qquad \textrm{for}\quad |\Phi - \Phi_1|\leq \delta.
\end{equation*}
Suppose that $\tilde{F}$ is small in the sense that for some positive constants $C_2$ and $\theta$ satisfying $C_0(C_1+C_2)\leq \theta<1$, we have
\begin{equation*}
|\tilde{F}(\Phi_1)|<\delta (1-\theta)/C_0 \quad \textrm{and} \quad |D\tilde{F}(\Phi)|\leq C_2 \qquad \textrm{for}\quad |\Phi - \Phi_1|\leq \delta.  
\end{equation*}
Then the equation $\Phi=F(\Phi)+\tilde{F}(\Phi)$ has a unique fixed point satisfying $|\Phi-\Phi_1|\leq \delta$, and moreover
\begin{equation}
|\Phi-\Phi_1|\leq (1-\theta)^{-1}C_0|\tilde{F}(\Phi_1)|<\delta.
\end{equation}
\end{lemma}

\subsection{Proof of the main result}
Here we finally give the proof of our main result Theorem~\ref{thm:mainresult}, by combining the Fourier multiplier estimate from Lemma \ref{lem:pseudoestimate},
the spectral information gathered in Lemma \ref{lem:specanalysis}, and the above version of the implicit function theorem.

\begin{proof}[Proof of Theorem~\ref{thm:mainresult}]
We first prove statements (a) and (b). Introduce the operators
\begin{equation*}
F(\Phi)=P^{(0)}N^{(0)}(\Phi), \qquad \tilde{F}(\Phi)=P^{(\eps)}N^{(\eps)}(\Phi)-P^{(0)}N^{(0)}(\Phi),
\end{equation*}
where the symbols of the operators $P^{(\eps)}$ and $P^{(0)}$ are given by \eref{equ:symbolpeps} and \eref{equ:symbolp0}, respectively. Moreover, recall that
\begin{equation*}
N^{(\eps)}(\Phi)=\frac{V'''(0)}{2}\Phi^2(1+\eta(\eps^2\Phi)), \qquad N^{(0)}(\Phi)=\frac{V'''(0)}{2}\Phi^2.
\end{equation*}
We now want to solve the fixed point equation \eref{equ:fixedpointphieps}, which is equivalent to the equation $\Phi=F(\Phi)+\tilde{F}(\Phi)$, by the implicit function theorem. Hence we prove that on the space $E_{2L_0}^1(\R)$ of even functions in $H_{2L_0}^1(\R)$ the operators $F$ and $\tilde{F}$ satisfy the criteria of Lemma~\ref{lem:IFT}. That $F$ fulfills the required properties has been shown in Lemma~\ref{lem:specanalysis}. It remains to establish the estimates on $\tilde{F}$ and $D\tilde{F}$. For that purpose write
\begin{eqnarray*}
\tilde{F}(\Phi)&=&P^{(\eps)}N^{(\eps)}(\Phi)-P^{(0)}N^{(0)}(\Phi)\\
&=&P^{(\eps)}\big(N^{(\eps)}(\Phi)-N^{(0)}(\Phi)\big)+\big( P^{(\eps)}-P^{(0)}\big)N^{(0)}(\Phi)\\
&=&P^{(\eps)}\frac{V'''(0)}{2}\Phi^2\eta(\eps^2\Phi)+\big( P^{(\eps)}-P^{(0)}\big)\frac{V'''(0)}{2}\Phi^2.
\end{eqnarray*}
Now choose $\delta>0$ sufficiently small such that the hypotheses of the Lemma~\ref{lem:IFT} are guaranteed. Since $\eta\in C^1$ with $\eta(0)=0$ we have $\|\eta(\eps^2\Phi)\|\leq C\eps^2$. Moreover, using Lemma \ref{lem:pseudoestimate} b) we deduce
\begin{equation*}
\| \tilde{F}\|_{\mathcal{L}(H^1_{2L_0})}\leq C\eps^2, \qquad \| D\tilde{F}\|_{\mathcal{L}(H^1_{2L_0})}\leq C\eps^2,
\end{equation*}
where $C$ is a constant independent of $\eps$. Then Lemma~\ref{lem:IFT} implies that there exists a unique solution $\Phi_1^{(\eps,k_0,L_0)}$ of \eref{equ:fixedpointphieps} satisfying
\begin{equation*}
\| \Phi^{(\eps,k_0,L_0)} -\Phi_1^{(k_0,L_0)}\|\leq C\eps^2,
\end{equation*}
where $\Phi_1^{(k_0,L_0)}$ is the fixed point solution of $\Phi=F(\Phi)$ and given by \eref{equ:cnoidalsolkL} with $k_0$ and $L_0$ fixed such that $c_{KdV}(k_0,L_0)=1$. With $r_c:=\eps^2\Phi^{(\eps,k_0,L_0)}(\eps .)$ this proves statements (a) and (b) of the theorem.

For proving the statement (c), we fix $\eps$ but reintroduce the parameter $c_{KdV}=4K^2(k)\sqrt{1-k^2+k^4}/L^2$ by assuming the parameters $k\in(0,1)$ and $L\in \R^+$ not to be fixed. That is, we take $c^2=V''(0)(1+\eps^2c_{KdV}/12)$. Then we define the map
\begin{equation*}
G(\Phi,k,L)=P^{(\eps)}N^{(\eps)}(\Phi), \qquad G:E^1_{2L}\times (0,1)\times \R^+ \longrightarrow E^1_{2L}.
\end{equation*}
Now observe that $(k,L)\mapsto c_{KdV}$ and thus $(k,L)\mapsto c$ is $C^1$. Thus the mapping
\begin{equation*}
(k,L)\mapsto p=\sinc^2(\,\cdot\, \frac{\pi}{2L})/(c^2-V''(0)\sinc^2(\,\cdot\, \frac{\pi}{2L}))
\end{equation*}
is smooth from $(0,1)\times \R^+$ into $L^{\infty}(\R)$. From this it follows that the map $G$ is $C^1$. For $(k_0,L_0)$ fixed such that $c_{KdV}=1$ we have shown that there is a fixed point $\Phi_1^{(\eps,k_0,L_0)}$ and from the existence proof it follows that $I-D_{\Phi}G(\Phi_1^{(\eps,k_0,L_0)},k_0,L_0)$ is an isomorphism on $E_{2L_0}^1$.

From the standard implicit function theorem \cite[Theorem 4B]{zeidler1986nonlinear} we can thus deduce the following: for some small enough neighbourhood of $(k_0,L_0)\in (0,1)\times \R^+$ there exists a $C^1$-curve $(k,L)\mapsto \Phi^{(\eps,k,L)}$ of solutions of \eref{equ:fixedpointphieps} such that $\Phi^{(\eps,k_0,L_0)}=\Phi_1^{(\eps,k_0,L_0)}$. Thus the mapping $(k,L)\mapsto c_{KdV}^{-1} \Phi^{(\eps,k,L)}(c_{KdV}^{-1/2}.)$ is $C^1$ using again that $(k,L)\mapsto c_{KdV}$ is $C^1$. Note that $c_{KdV}^{-1}\Phi^{(\eps,k,L)}(c_{KdV}^{-1/2}.)$ is a solution of $\Phi=P^{(\eps_1)}N^{(\eps_1)}(\Phi)$ with $\eps_1=\sqrt{c_{KdV}}\eps$ and $c^2=V''(0)(1+\eps_1^2/12)$. Moreover, if $\eps_1$ is sufficiently small and $(k,L)$ in a small enough neighbourhood of $(k_0,L_0)$ then $\| c_{KdV}^{-1} \Phi^{(\eps,k,L)}(c_{KdV}^{-1/2}.) - \Phi_1^{(\eps,k_0,L_0)}\|<\delta$ and so $c_{KdV}^{-1} \Phi^{(\eps,k,L)}(c_{KdV}^{-1/2}.)=\Phi^{(\eps_1,k,L)}$. It follows that $(k,L)\mapsto \Phi^{(\eps_1(
 k,L),k,L)}$ and hence $c=(V''(0)(1+\eps_1^2/12))^{1/2}\mapsto r_c=\eps_1^2\Phi^{(\eps_1,k,L)}(\eps_1 .)$ are $C^1$ maps into $H^1_{2L}$.
\end{proof}

\appendix

\section{Results from the theory on the Lam\'e equation} \label{app:Lame}

The spectral problem for $L$ defined by \eref{equ:linearizationop} which reads 
\begin{equation} \label{equ:specprobL}
L\psi(\xi)=\lambda \psi (\xi) \qquad \textrm{on} \quad H^1_{2L}(\R),
\end{equation}
is equivalent to the following problem
\begin{equation} \label{equ:specprobHcn}
H_{cn}^{\lambda}\psi(y)=h_{cn}^{\lambda}\psi(y) \qquad \textrm{on} \quad H^1_{2K}(\R),
\end{equation}
where the transformed coordinate, operator, and eigenvalue are 
\begin{eqnarray}
y&=&\frac{K}{L}\xi,\\
H_{cn}^{\lambda}&=&-\partial_y^2-\frac{12}{\lambda}k^2\cn^2(y;k^2),\\
h_{cn}^{\lambda}&=&\frac{1}{\lambda}(4-8k^2)+4\Big( \frac{1}{\lambda}-1\Big)\sqrt{1-k^2+k^4}.
\end{eqnarray} 

Using the relation $\cn^2+\sn^2=1$ this spectral problem transforms into
\begin{equation} \label{equ:specprobHsn}
H_{sn}^{\lambda}\psi(y)=h_{sn}^{\lambda}\psi(y) \qquad \textrm{on} \quad H^1_{2K}(\R),
\end{equation}
where the operator and eigenvalue are
\begin{eqnarray}
H_{sn}^{\lambda}&=&-\partial_y^2+\frac{12}{\lambda}k^2\sn^2(y;k^2), \label{equ:Hsnlambda}\\
h_{sn}^{\lambda}&=&h_{cn}^{\lambda}+\frac{12}{\lambda}k^2 \label{equ:eigenvaluehsn}.
\end{eqnarray} 
The general eigenvalue problem for the operator $H_{sn}^{\lambda}$ is connected to the Jacobian form of the Lam\'e equation, namely 
\begin{equation} \label{equ:lameequ}
-\psi''(y)+n(n+1)k^2\sn^2(y;k^2)\psi(y)=E\psi(y).
\end{equation}
For $n\in \N$ the solutions of this equation are meromorphic in the whole complex plane, while for $n$ other than an integer they are multi-valued. We will use results from the theory on the Lam\'e equation \cite{ince1940aperiodic,arscott1964periodic} combined with well-known facts from Floquet theory \cite{magnus2004hill}.

We start by recalling the basic facts from Floquet theory. Let us consider the spectral problem \eref{equ:lameequ} together with the boundary condition
\begin{equation*}
\psi(y+2K)=\E ^{\I 2K \kappa}\psi(y):=\sigma \psi(y),
\end{equation*}
where $\kappa$ is referred to as the momentum and $\sigma$ the Floquet multiplier. For any $\kappa \in \R$, i.e.~$|\sigma |=1$, this defines a self-adjoint boundary value problem. The spectrum of the \emph{periodic} ($\sigma=1$) and \emph{semi-periodic} ($\sigma=-1$) problem consists of infinitely many discrete eigenvalues which only accumulate at $+\infty$ and the corresponding eigenfunctions are $2K$- and $4K$-periodic, respectively. Let $E_0^+\leq E_1^+\leq E_2^+\leq \dots$ denote the eigenvalues and $\psi_0^+,\psi_1^+,\psi_2^+,\dots$ the $2K$-periodic eigenfunctions of the periodic problem. Similarly, $E_0^-\leq E_1^-\leq E_2^-\leq\dots$ and $\psi_0^-,\psi_1^-,\psi_2^-,\dots$ are the eigenvalues and $4K$-periodic eigenfunctions of the semi-periodic problem. We have that
\begin{itemize}
\item[(i)] $\psi_0^+$ has no zeros and $\psi_0^-$ exactly one zero in $[0,2K(k)]$,
\item[(ii)] $\psi_{2i-1}^+$ and $\psi_{2i}^+$ have exactly $2i$ zeros in the interval $[0,2K)$, for $i\geq 1$, and
\item[(iii)] $\psi_{2i-1}^-$ and $\psi_{2i}^-$ have exactly $2i-1$ zeros in the interval $[0,2K)$, for $i\geq 1$.
\end{itemize}
Then by oscillation theory we have
\begin{equation*}
E_0^+<E_0^-\leq E_1^-<E_1^+\leq E_2^+<E_2^-\leq E_3^- <\dots
\end{equation*}
The bands 
\begin{equation*}
(E_0^+,E_0^-),\ (E_1^-,E_1^+),\ (E_2^+,E_2^-),\dots
\end{equation*}
are the \emph{stable bands}. For each value in these bands we have $\kappa \in \R$ and hence $|\sigma|=1$, i.e.~only bounded but non-periodic solutions exist. Inbetween these bands we have gaps 
\begin{equation*}
(E_0^+,E_0^-),\ (E_1^-,E_1^+),\ (E_2^+,E_2^-),\dots
\end{equation*}
which are called \emph{bands of instability}. For values of $E$ lying in these gaps only unbounded non-periodic solutions exist, since the multiplier $|\sigma|\neq 1$ and the momentum $\kappa$ can not be real.

The solutions of the Lam\'e equation \eref{equ:lameequ} are given by
\begin{equation*}
\psi_{\pm}(y)=\prod_{j=1}^{n}\frac{H(y\pm\alpha_j)}{\Theta(y)}\exp\left\{ \mp yZ(\alpha_j)\right\},
\end{equation*}
for certain constants $\alpha_1,\dots,\alpha_n$, see \cite{whittaker1962course}. Here $H$ denotes the Jacobian eta, $\Theta$ the Jacobian theta, and $Z$ the Jacobian zeta function \cite{byrd1971handbook}. At a band edge the eigenfunctions correspond to polynomials in the Jacobian elliptic functions $\cn$, $\sn$, $\dn$ and are called \emph{Lam\'e polynomials}.

It is known \cite{ince1940aperiodic}, \cite[Theorem~7.8]{magnus2004hill} that if and only if $n$ is an integer, $2K$- or $4K$-periodic solutions can coexist, that is the spectrum has finite-band structure (since two band edges collide) with $n+1$ bands (if $n$ is non-negative). The $2n+1$ band edges correspond to simple eigenvalues of the periodic or semi-periodic problem.

Let us consider the \emph{periodic} problem 
\begin{equation} \label{equ:specprobP}
H_{sn}^n\psi(y)=h\psi(y), \qquad \psi(0)=\psi(2K), \quad \psi'(0)=\psi'(2K),
\end{equation}
where $H_{sn}^n$ equals the operator given by \eref{equ:Hsnlambda} if we set $12/\lambda=n(n+1)$.
Ince \cite{ince1940aperiodic} studied the Lam\'e equation for $n\in \R$, hence not necessarily an integer. From his results it follows that there exist $2K$-periodic eigenfunctions of $H_{sn}^n$ with corresponding eigenvalue $h$ of the specific form \eref{equ:eigenvaluehsn}, that is in terms of $n$
\begin{equation} \label{equ:eigenvaluehsnn}
h^n_{sn}=\frac{1}{3} \Big( n(n+1)(1+k^2)+4(n(n+1)-3)\sqrt{1-k^2+k^4}\Big),
\end{equation}
if and only if $n=2$ ($\lambda=2$) or $n=3$ ($\lambda=1$). Only in these cases Ince's continued fraction representations for $h$ terminate and give rise to solutions given by the finite form \eref{equ:eigenvaluehsnn}. In either cases, $n=2$ and $n=3$, the concerned eigenvalues $h_{sn}^{2}$ and $h_{sn}^{3}$ are simple and correspond to a band edge of the finite band spectrum of $H_{sn}^2$ and $H_{sn}^3$, respectively.
Note that an eigenvalue $\lambda$ of the spectral problem \eref{equ:specprobL} is simple if and only if the corresponding eigenvalue $h_{sn}^{\lambda}$ of \eref{equ:specprobHsn} is simple.

{\bf Case $\lambda=1$}:

In this case the spectrum of the operator $H_{sn}^3$ consists of the four spectral bands
\begin{equation*}
(E_0^+,E_0^-), (E_1^-,E_1^+), (E_2^-,E_2^+), (E_3^-,\infty),
\end{equation*}
where the $E_i^+$'s are the eigenvalues corresponding to the periodic problem, and the $E_i^-$'s the ones corresponding to the semi-periodic problem. Hence the first three eigenvalues $E_0^+,E_1^+,E_2^+$ of \eref{equ:specprobP} are simple and all other eigenvalues $E_i^+$ with $i\geq 3$ are double. From \cite{ince1940aperiodic} or \cite{arscott1964periodic} (the eigenvalues and eigenfunctions are calculated explicitly for the cases $n=1,2,3$) we get results summarized in Table~\ref{tab:specproperties1}.

\begin{table}[h!]
\begin{center}
\begin{tabular}{|l|l|c|c|}
\hline
\multicolumn{4}{|c|}{Case $n=3$ ($\lambda=1$)}\\
\hline
eigenvalue & eigenfunction & period & parity \\
\hline
$E_0^+=2+5k^2-2\sqrt{1-k^2+4k^4}$ & $\psi_0^+(y)=\dn y[1-(E_0^+-k^2)\sn^2y]$ & $2K$ & even\\
$E_0^-=5+2k^2-2\sqrt{4-k^2+k^4}$ & $\psi_0^-(y)=\cn y[1-(E_0^--1)\sn^2y]$ & $4K$ & even\\
$E_1^-=5+5k^2-2\sqrt{4-7k^2+4k^4}$ & $\psi_1^-(y)=3\sn y-(E_1^--1-k^2)\sn^3y$ & $4K$ & odd\\
$E_1^+=4+4k^2$ & $\psi_1^+(y)=\sn y\ \cn y\ \dn y$ & $2K$ & odd\\
$E_2^+=2+5k^2+2\sqrt{1-k^2+4k^4}$ & $\psi_2^+(y)=\dn y[1-(E_2^+-k^2)\sn^2y]$ & $2K$ & even\\
$E_2^-=5+2k^2+2\sqrt{4-k^2+k^4}$ & $\psi_2^-(y)=\cn y[1-(E_2^--1)\sn^2y]$ & $4K$ & even\\
$E_3^-=5+5k^2+2\sqrt{4-7k^2+4k^4}$ & $\psi_3^-(y)=3\sn y-(E_3^--1-k^2)\sn^3y$ & $4K$ & odd\\
\hline 
\end{tabular}
\end{center}
\caption{The spectral quantities of the spectral problem of the operator $H_{sn}^3=-\partial^2_y+12k^2\sn^2(y;k^2)$ with periodic or semi-periodic boundary conditions. The eigenvalues correspond to the band edges of the finite band spectrum. The associated eigenfunctions have period $2K$ or $4K$ and are polynomials in the Jacobian elliptic functions $\sn(y;k^2)$, $\cn(y;k^2)$, and $\dn(y;k^2)$.\label{tab:specproperties1}}
\end{table}

The eigenvalue $E_1^+=4+4k^2$ is indeed of the form \eref{equ:eigenvaluehsnn} for $n=3$ and the corresponding eigenfunction has period $2K$ and is given by $\psi_1^+(y)=\sn y\ \cn y\ \dn y$. Thus tracing back definitions we get the following result on the spectral problem \eref{equ:specprobL}: the operator $L$ has the simple eigenvalue $\lambda=1$ with the $2L$-periodic odd eigenfunction $\psi(\xi)=\sn(K\xi/L;k^2)\cn(K\xi/L;k^2)\dn(K\xi/L;k^2)$. Indeed, this gives the same result as derived in Lemma~\ref{lem:specanalysis}. To see this note that $(\cn y)'=-\sn y\ \dn y$ and hence $\frac{d}{d\xi}\Phi_1^{(k,L)}\propto \psi(\xi)$. 

The eigenvalues of the periodic and semi-periodic spectral problem listed in Table~\ref{tab:specproperties1} depend on the elliptic modulus $k^2$. This dependence is visualized in Figure~\ref{fig:bandslambda1}. In the limiting case $k^2\to 1$ we have $\cn y \approx \sech y$. The finite band structure arising from the periodic potential turns into a purely discrete spectrum. The corresponding eigenfunctions are given by the associated Legendre polynomials in $\tanh y$. In the case $k^2\to 0$ we have $\cn y\approx \cos y$ and the Lam\'e equation turns into Mathieu's equation. The eigenfunctions are then polynomials in the Mathieu cosine and sine function.

\begin{figure}[h!]
\centering
\includegraphics[width=.6\textwidth]{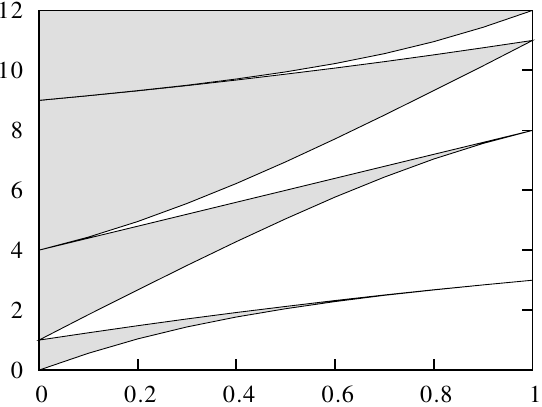}
\caption{The finite band structure of the periodic operator $H_{sn}^3$. The band edges depend on the elliptic modulus $k^2\in (0,1)$ as given in Table~\ref{tab:specproperties1}. \label{fig:bandslambda1}}
\end{figure}

{\bf Case $\lambda=2$}:

In this case the spectrum of the operator $H_{sn}^2$ consists of the three spectral bands
\begin{equation*}
(E_0^+,E_0^-), (E_1^-,E_1^+), (E_2^-,\infty),
\end{equation*}
where the $E_i^+$'s are the eigenvalues corresponding to the periodic problem, and the $E_i^-$'s the one corresponding to the semi-periodic problem. Hence the first two eigenvalues $E_0^+$ and $E_1^+$ of \eref{equ:specprobP} are simple and all other eigenvalues $E_i^+$ with $i\geq 2$ are double. From \cite{ince1940aperiodic} or \cite{arscott1964periodic} we get the results summarized in Table~\ref{tab:specproperties2}.

\begin{table}[h!]
\begin{center}
\begin{tabular}{|l|l|c|c|}
\hline
\multicolumn{4}{|c|}{Case $n=2$ ($\lambda=2$)}\\
\hline
eigenvalue & eigenfunction & period & parity \\
\hline
$E_0^+=2+2k^2-2\sqrt{1-k^2+k^4}$ & $\psi_0^+(y)=2-E_0^+\sn^2y$ & $2K$ & even\\
$E_0^-=1+k^2$ & $\psi_0^-(y)=\cn y\ \dn y$ & $4K$ & even\\
$E_1^-=1+4k^2$ & $\psi_1^-(y)=\sn y\ \dn y$ & $4K$ & odd\\
$E_1^+=4+k^2$ & $\psi_1^+(y)=\sn y\ \cn y$ & $2K$ & odd\\
$E_2^+=2+2k^2+2\sqrt{1-k^2+k^4}$ & $\psi_2^+(y)=\dn y[1-(E_2^+-k^2)\sn^2y]$ & $2K$ & even\\
\hline 
\end{tabular}
\end{center}
\caption{The spectral quantities of the spectral problem of the operator $H_{sn}^2=-\partial^2_y+6k^2\sn^2(y;k^2)$ with periodic or semi-periodic boundary conditions. The eigenvalues correspond to the band edges of the finite band spectrum. The associated eigenfunctions have period $2K$ or $4K$ and are polynomials in the Jacobian elliptic functions $\sn(y;k^2)$, $\cn(y;k^2)$, and $\dn(y;k^2)$.\label{tab:specproperties2}}
\end{table}
The eigenvalue $E_0^+=2+2k^2-2\sqrt{1-k^2+k^4}$ is indeed of the form \eref{equ:eigenvaluehsnn} for $n=2$ and the corresponding eigenfunction has period $2K$ and is given by $\psi_0^+(y)=2-E_0^+\sn^2y$. Thus tracing back definitions we get the following result on the spectral problem \eref{equ:specprobL}: the operator $L$ has the simple eigenvalue $\lambda=2$ with the $2L$-periodic even eigenfunction $\psi(\xi)=2-[1+k^2-\sqrt{1-k^2+k^4}]\sn^2(K\xi/L;k^2)$. Using the relation $\sn^2y=1-\cn^2$ we get $\Phi_1^{(k,L)}(\xi)\propto \psi(\xi)$, verifying our result derived in Lemma~\ref{lem:specanalysis}.

\section{Quasi-periodic KdV solutions} \label{app:periodicKdVsol}

In this appendix we give an overview of the beautiful theory of quasi-periodic KdV solutions. 

The inverse scattering transform developed by Gardner et al.~\cite{gardner1974korteweg} allows to solve the Cauchy problem for the KdV equation
\begin{equation} \label{equ:KdVZK1}
u_t+6uu_x+u_{xxx}=0.
\end{equation}
Unfortunately, this method is restricted to rapidly decreasing initial data. Due to the lack of methods for solving the inverse spectral problem for $H$ with periodic potential it had not been immediately clear how to adopt the inverse spectral method to the periodic case.

A breakthrough in this problem was achieved independently by Dubrovin \cite{dubrovin1975periodic,dubrovin1975inverse}, Its and Matveev \cite{its1975hill,its1975schroedinger}, and McKean and van Moerbeke \cite{mckean1975spectrum}, generalizing an earlier result by Akhiezer \cite{akhiezer1961theory}. 
They proposed a scheme to construct so-called finite-gap KdV solutions from given spectral data. 

To describe this method in a little more detail, consider the Lax operator with a periodic potential, i.e.,
\begin{equation*}
H=-\frac{d^2}{dx^2}+u(x,0), \qquad u(x+2L,0)=u(x,0).
\end{equation*}
Assume the absolutely continuous part of the spectrum of $H$ has finite-band structure, i.e.~it consists of $N+1$ intervals $[E_{2j-1},E_{2j}]\in \R$, $j=1,\dots,N+1$ with $E_{2N+2}=\infty$. Then consider the Dirichlet problem 
\begin{equation*}
H\psi=\mu \psi, \qquad \psi(0,0)=\psi(2L,0)=0.
\end{equation*}
We refer to the spectrum of this problem as auxiliary spectrum consisting of a set of discrete points $\{\mu_j(0,0)\}_{j=1}^{\infty}$. All but $N$ points of this set lie in the infinite interval $[E_{2N+1},\infty)$ of the continuous spectrum. Each of the $N$ remaining points is contained in one of the $N$ spectral gaps. Finally, the solution $u(x,t)$ is given by the trace formula
\begin{equation} \label{equ:traceformula}
u(x,t)=\sum_{i=1}^{2N+1} E_i-2\sum_{j=1}^N\mu_j(x,t).
\end{equation}
Here the functions $\mu_j(x,t)$, $j=1,\dots,N$ denote the space-time evolution of the auxiliary spectrum. These quantities live on the Riemann surface associated with the curve
\begin{equation} \label{equ:riemannsurface}
R:\ y^2=-\prod_{j=1}^{2N+1}(z-E_j),
\end{equation}
where each of them lies in a spectral gap, i.e.~$\mu_j(x,t)\in (E_{2j},E_{2j+1})$, $j=1,\dots,N$ and oscillates between these band edges as $x$ varies. More precisely, the dynamics in $x$ and $t$ is given by the Dubrovin equations
\begin{equation}
\mu_j'=\frac{\pm 2 y(\mu_j)}{\prod_{k\neq j}(\mu_j+\mu_k)},\quad \dot{\mu}_j=2\mu_j'\big( \sum_{i}E_i-2\sum_{k\neq j}\mu_k\big),
\end{equation}
where the prime and dot denote the derivative with respect to the space and time variable, resepctively.
Solving these equations is a matter of applying the Jacobi inversion problem to the so-called Abel map, that is a map from the Riemann surface $R$ to an $N$-torus called the Jacobi variety. In the case $N=1$, finding a periodic KdV solution is linked to the classical inversion problem for the elliptic integral, cf.~Section~\ref{sec:KdVsolutions}. 

The expression \eref{equ:traceformula} for the solutions of the KdV equation \eref{equ:KdVZK1}, called finite-gap KdV solutions, is finally given by
\begin{eqnarray} 
&&u(x,t)=c+2\frac{\partial^2}{\partial x^2}\ln \Theta(\mathbf{\xi}), \label{equ:itsmatveevformula1}\\
&&\Theta(\mathbf{\xi};B)=\sum_{\mathbf{k}\in\Z^N}\exp \left\{ 2\pi\I \langle \mathbf{k},\mathbf{\xi}\rangle+\pi \I \langle \mathbf{k},B\mathbf{k}\rangle\right\} ,\nonumber
\end{eqnarray}
where $\langle .,.\rangle$ denotes the scalar product in $\R^N$ and $c$ a constant. Moreover, $B$ is the symmetric period matrix of the Riemann surface $R$ of the curve \eref{equ:riemannsurface} and $\mathbf{\xi}=(\xi_1,\dots,\xi_N)$ denotes the vector of the generalized phases $\xi_j=K_jx-\omega_jt+\phi_j$, $j=1,\dots,N$. The $K_j$'s and $\omega_j$'s are completely determined by the continuous and auxiliary spectrum and can be computed in terms of Abelian integrals. Note that the Riemann $\Theta$-function fulfills the property 
\begin{equation*}
\Theta(\mathbf{\xi}+\mathbf{m}+B\mathbf{n};B)=\exp \left\{ -2\pi \I \langle \mathbf{n},\mathbf{\xi}\rangle-\pi \I \langle \mathbf{n},B\mathbf{n}\rangle\right\} \Theta (\mathbf{\xi};B),
\end{equation*}
for all $\mathbf{n},\mathbf{m}\in \Z^N$ and symmetric $N\times N$ matrices $B$.

In the one-gap case ($N=1$), the Riemann $\Theta$-function reduces to one of the Jacobi $\theta$-functions which are linked to the Jacobi elliptic functions $\cn$, $\sn$, and $\dn$ \cite{byrd1971handbook}. The solution \eref{equ:itsmatveevformula1} is then given by
\begin{equation} \label{equ:cnoidalwave1}
u(x,t)=E_2+(E_3-E_2)\cn^2\Big( \sqrt{\frac{E_3-E_1}{2}}\big(x-2(E_1+E_2+E_3)t\big);k^2\Big),
\end{equation}
where $E_1<E_2<E_3$ and $k^2=(E_3-E_2)/(E_3-E_1)$ is the elliptic modulus. Equation~\eref{equ:cnoidalwave1} is the so-called cnoidal wave which describes a periodic travelling wave solution of \eref{equ:KdVZK1} and has been already known to Korteweg and de~Vries in 1895 \cite{korteweg1895change}. The infinite period limit, $E_1\to E_2$ ($k^2\to1$), gives the single soliton 
\begin{equation} \label{equ:1soliton}
u(x,t)=\beta \sech^2\Big( \sqrt{\frac{\beta}{2}}\big(x-\beta t\big)\Big),
\end{equation}
where we have set $\beta=E_3-E_2$ and imposed the condition $u\to 0$ as $x\to \pm \infty$.

Note that any finite-gap solution of the form \eref{equ:itsmatveevformula1} is an Abelian function\footnote{An Abelian function is a meromorphic function in $N$ complex variables with $2N$ independent periods, the so-called fundamental periods. It is usually defined as a function on an Abelian variety.}, which constitutes a generalization of elliptic functions to $N$ variables. Thus, in the special case $N=1$ the solution is elliptic and referred to as \emph{cnoidal wave} or \emph{Lam\'e potential}.

The solution \eref{equ:itsmatveevformula1} is periodic in $x$ (since the initial condition is), but in general is only \emph{quasi-periodic}\footnote{A finite-gap solution is called periodic in $x$ with period $2L$ if the wave numbers $K_1,\dots,K_N$ are commensurable, i.e.~there exist $N$ integers $n_1,\dots,n_N$ such that $2LK_j=n_j$, $j=1,\dots,N$. Otherwise it is called \emph{quasi-periodic}.} in $t$. In general, however, any finite-gap KdV solution is quasi-periodic in both $x$ and $t$, since their construction via algebro-geometric methods \cite{belokolos1994algebro,gesztesy2003sea} allows to choose the quantities $\{E_j\}_{j=1}^{2N+1}\in \R$ and $\{\mu_j(0,0)\}_{j=1}^{N}$ arbitrarily such that the $K_j$'s and $\omega_j$'s are in general incommensurate quantities. More precisely, these methods use results on Abelian functions and Riemann surfaces. A key role
  is played by the eigenfunctions of the spectral problem for the Lax operator $H$. They turn out to correspond to the two distinct branches of one single-valued function, the so-called Baker-Akhiezer function, defined on a two-sheeted Riemann surface given by \eref{equ:riemannsurface}. 

Note that in the case of the the Toda lattice, where the interaction potential is given by $V(r)=\alpha (\E^{-\beta r}+\beta r-1)$ and the system is integrable, there have been constructed exact cnoidal wave solutions \cite{toda1981theory}. For the construction of all quasi-periodic finite-gap solutions for the Toda lattice we refer to \cite{teschl2000joa} or \cite{gesztesy2008sea}.

Finally, we would like to point out two major points which emphasize the significance of finite-gap solutions for understanding KdV dynamics:
\begin{itemize}
\item[(i)]
The finite-gap solution \eref{equ:itsmatveevformula1} degenerates to a multi-soliton solution if the band edges of the continuous spectrum and thus the branch points of the corresponding Riemann surface $R$ collide, i.e.~$E_{2j}\to E_{2j+1}$, $j=0,1,\dots,N-1$. In this case the Riemann theta function $\Theta$ reduces to a determinant with exponential entries. In this way the finite-gap KdV solutions \eref{equ:itsmatveevformula1} corresponding to $N$ gaps in the continuous spectrum constitute a natural (quasi-)periodic extension of the KdV $N$-soliton solutions.
\item[(ii)] The KdV equation has \emph{almost-periodic\footnote{A function $f$ is called \emph{almost periodic} in the sense of Bohr if for each $\eps>0$ there exists a relatively dense set of $\eps$-almost-periods, i.e.~translations $\tau$ with the property~$\sup_{x}(f(x+\tau),f(x))\leq\eps$. Alternatively, there is the \emph{Bochner criterion} for almost periodicity: a continuous function $f$ is almost periodic if and only if the family of functions $H=\{ f^h\}=\{f(x+h)\}$, $-\infty<h<\infty$, is compact in the set of continuous bounded functions.} infinite-gap} solutions corresponding to an infinite number of gaps in the spectrum \cite{levitan1982almost}. The set of all quasi-periodic finite-gap solutions (which also contains all periodic finite-gap solutions) is dense in the former set. This means, in particular, that any periodic KdV solution can be approximated by finite-gap KdV solutions \cite{levitan1983approximation}. These considerations concern the periodicity in the spatial variable, but what about the behaviour in time? It is known \cite{mckean1976hill,lax1976almost} that all KdV solutions which are periodic in $x$ are almost periodic in $t$.  

\end{itemize}

\ack
We thank Gerald Teschl for helpful comments, especially concerning relevant literature on the Korteweg-de Vries equation and its solutions. \\
A.M.~gratefully acknowledges the financial support by the Austrian Science Fund (FWF) under grant J3143.


\end{document}